\documentclass[onecolumn, 12pt]{IEEEtran}
\usepackage{amsmath}
\usepackage{amssymb}
\usepackage{amsfonts}
\usepackage{graphicx}
\usepackage{epsfig}
\usepackage{subfigure}
\usepackage{psfrag}
\usepackage{cite}
\usepackage{latexsym}
\usepackage{url}
\usepackage{color}

\begin{document}
\title{Multiuser MIMO Wireless Energy Transfer With Coexisting Opportunistic Communication
\footnote{The authors are with the Department of Electrical and Computer Engineering,
National University of Singapore, Singapore (e-mail: \{elexjie, bsz, elezhang\}@nus.edu.sg).}}

\author{Jie~Xu,~Suzhi~Bi,~and~Rui~Zhang}

\maketitle
\begin{abstract}
This letter considers spectrum sharing between a primary multiuser multiple-input multiple-output (MIMO) wireless energy transfer (WET) system and a coexisting secondary point-to-point MIMO wireless information transmission (WIT) system, where WET generates interference to WIT and degrades its throughput performance. We show that due to the interference, the WIT system suffers from a loss of the degrees of freedom (DoF) proportional to the number of energy beams sent by the energy transmitter (ET), which, in general, needs to be larger than one in order to optimize the multiuser WET with user fairness consideration. To minimize the DoF loss in WIT, we further propose a new single-beam energy transmission scheme based on the principle of time sharing, where the ET transmits one of the optimal energy beams at each time. This new scheme achieves the same optimal performance for the WET system, and minimizes the impact of its interference to the WIT system.
\end{abstract}

\setlength{\baselineskip}{1.28\baselineskip}
\newtheorem{definition}{\underline{Definition}}[section]
\newtheorem{fact}{Fact}
\newtheorem{assumption}{Assumption}
\newtheorem{theorem}{\underline{Theorem}}[section]
\newtheorem{lemma}{\underline{Lemma}}[section]
\newtheorem{corollary}{\underline{Corollary}}[section]
\newtheorem{proposition}{\underline{Proposition}}[section]
\newtheorem{example}{\underline{Example}}[section]
\newtheorem{remark}{\underline{Remark}}[section]
\newtheorem{algorithm}{\underline{Algorithm}}[section]
\newcommand{\mv}[1]{\mbox{\boldmath{$ #1 $}}}

\section{Introduction}
Radio frequency (RF) signal enabled wireless energy transfer (WET) has become an attractive technology to provide convenient and perpetual power supply to future energy-constrained wireless networks \cite{Bi:Overview}. The natural integration of WET and conventional wireless information transmission (WIT) systems has spurred many new wireless design paradigms that jointly investigate WET and WIT. For example, simultaneous wireless information and power transfer (SWIPT) (see, e.g., \cite{ZhangHo2013}) and wireless powered communication network (WPCN) (see, e.g., \cite{JuZhang2014}) have been proposed to enable simultaneous RF energy harvesting and information reception/transmission for wireless devices.

Instead of considering fully coordinated WET and WIT within one system as in SWIPT and WPCN, in this letter, we study a new yet practical scenario with two WET and WIT systems operating separately in the same geographical area. In particular, we consider spectrum sharing between the two systems for improving the spectrum utilization efficiency. In such a scenario, only {\it one-way} interference occurs from WET to WIT, since the radio signal from WIT to WET is a useful energy source for energy receivers' (ERs') RF energy harvesting (rather than undesired interference). This is in sharp contrast to the conventional {\it two-way} interference in spectrum sharing between different WIT systems \cite{RZhang2008Exploiting}. In this one-way interference setup, the WIT system needs to communicate opportunistically  subject to the interference from the WET system.

In this letter, we investigate the optimal energy and information signals design for WET and WIT coexisting within the same spectrum. For the purpose of exposition, we consider a primary multiuser WET and a secondary point-to-point WIT systems, where multi-antenna or multiple-input multiple-output (MIMO) technique is exploited at both systems for improving energy transfer efficiency and communication data rate, respectively. Under this setup, we first consider the WET system, which optimizes the transmit beamforming at energy transmitter (ET) to maximize the transferred energy to all ERs subject to energy fairness constraints among them. It is revealed that in general more than one transmit energy beams are needed to achieve the optimality. Due to the interference from the ET, we then show that the WIT system suffers from a loss of the degrees of freedom (DoF) proportional to the number of energy beams sent by the ET, which, thus, is generally larger than one. To minimize the DoF loss in WIT, we propose a new single-beam energy transmission scheme based on the principle of time sharing, where the ET transmits one of the above optimal energy beams at each time. This new scheme minimizes the impact of its interference to the WIT system, and achieves the same optimal performance for the WET system as the optimal multi-beam scheme.

\section{System Model}

\begin{figure}
\centering
 \epsfxsize=1\linewidth
    \includegraphics[width=9.5cm]{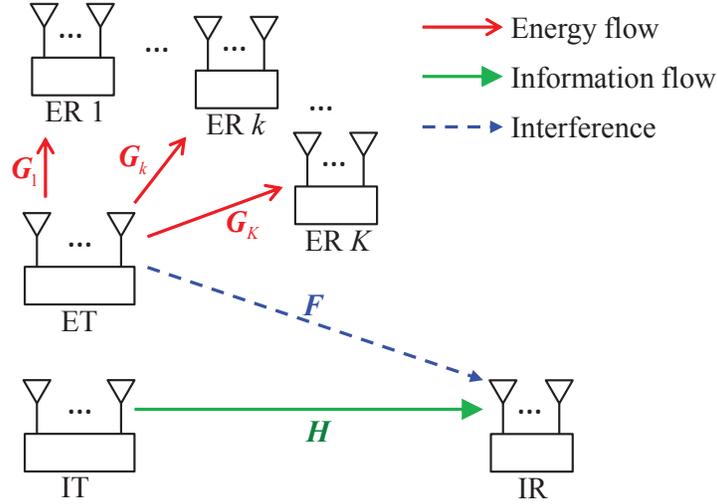}
\caption{Spectrum sharing between a (primary) multiuser MIMO WET and a (secondary) point-to-point MIMO WIT system with one-way interference from WET to WIT.}\label{fig:systemmodel}
\end{figure}

This letter considers a primary multiuser MIMO WET system and a coexisting secondary point-to-point MIMO WIT system as shown in Fig. \ref{fig:systemmodel}, where the two systems operate over the same transmit spectrum. There are one ET with $M_E$ antennas and $K$ ERs each with $N_E$ antennas in the WET system, as well as one information transmitter (IT) with $M_I$ antennas and one information receiver (IR) with $N_I$ antennas in the WIT system. We consider a quasi-static flat-fading channel model and a block-based energy/information transmission, where wireless channels remain constant over each transmission block with a length of $T > 0$. In addition, we assume perfect local channel state information (CSI) at the two systems, that is, the ET has the perfect CSI to all the ERs, while the IT and the IR accurately know the CSI between them.

First, we consider the energy/information transmission at the two systems. Let the transmit energy and information signals at the ET and the IT be denoted by $\mv x_E  \in \mathbb{C}^{M_E \times 1}$ and $\mv x_I \in \mathbb{C}^{M_I \times 1}$, and the corresponding transmit energy and information covariance matrices by $\mv S_E = \mathbb{E}\left(\mv x_E \mv x_E^H\right)$ and $\mv S_I = \mathbb{E}\left(\mv x_I \mv x_I^H\right)$, respectively. Note that given $\mv S_E$, $d_E = \mathrm{rank}(\mv S_E)$ in fact specifies the number of energy beams that are spatially transmitted \cite{LiuZhangChia}. In addition, we assume that the maximum transmit sum-power at the ET (the IT) is denoted by $P_E > 0$ ($P_I > 0$). Then we have $\mathbb{E}\left(\|\mv x_E\|^2\right) = \mathrm{tr}(\mv S_E) \le P_E$ and $\mathbb{E}\left(\|\mv x_I\|^2\right) = \mathrm{tr}(\mv S_I) \le P_I$.

Next, consider the energy harvesting at ERs. Due to the broadcast nature of radio signal, each ER  can harvest the energy carried by both the energy signal $\mv x_E$ from the ET and the information signal $\mv x_I$ from the IT. Because an IT often covers a large area (e.g., 100 meters radius) while an ET is only used for short range power transfer (e.g., a couple of meters range), in practice the distance from an ER to an IT is often much larger than that to an ET. As a result, the harvested energy from $\mv x_I$ is normally much weaker than that from $\mv x_E$, and thus could be safely omitted without compromising the performance. Let the MIMO channel matrix from the ET to each ER $k$ be denoted by $\mv G_k \in \mathbb{C}^{N_E \times M_E}, k \in \{1,\ldots,K\}$. Then the energy harvested by ER $k$ over the whole block is expressed as \cite{ZhangHo2013}
\begin{align}\label{eqn:energy}
Q_k\left(\mv S_E\right) =\eta T \mathrm{tr}\left(\mv G_k^H\mv G_k\mv S_E\right),
\end{align}
where $0<\eta\le 1$ denotes the energy harvesting efficiency at each ER. Since $\eta$ is a constant, we normalize it as $\eta=1$ in the sequel of this paper unless otherwise stated.

Finally, we consider the information reception at the IR. Let the MIMO channel matrices from the ET and the IT to the IR be denoted by $\mv F \in \mathbb{C}^{N_I \times M_E}$ and $\mv H \in \mathbb{C}^{N_I \times M_I}$, respectively. It is assumed that $\mv F$ and $\mv H$ are both of full-rank, i.e., $\mathrm{rank}(\mv F) = \min(N_I,M_E)$ and $\mathrm{rank}(\mv H) = \min(N_I,M_I)$, and all channels $\mv F$, $\mv H$, and $\mv G_k$'s are independently generated. Then the received signal at the IR is expressed by $\mv y = \mv H \mv x_I + \mv F \mv x_E + \mv n$, where $\mv H \mv x_I$ is the desired information signal sent from the IT, $\mv F \mv x_E$ is the co-channel interference caused by the energy signal transmitted from the ET, and $\mv n$ denotes the additive white Gaussian noise (AWGN) at the IR, which is a circularly symmetric complex Gaussian (CSCG) random vector with zero mean and covariance matrix $\sigma^2\mv I$, i.e., $\mv n \sim \mathcal{CN}(\mv 0,\sigma^2\mv I)$, with $\sigma^2 > 0$ denoting the noise power. Accordingly, the interference-plus-noise covariance matrix at the IR is given by
$\mathbb{E}\left((\mv F \mv x_E + \mv n)(\mv F \mv x_E + \mv n)^H\right) = \mv F\mv S_E\mv F^H + \sigma^2 \mv I.$ As a result, by assuming Gaussian signalling at the IT, the achievable rate at the IR (in bps/Hz) is given by
\begin{align}\label{eqn:info:rate}
R&\left(\mv S_E,\mv S_I\right) \nonumber \\&= \log_2 \det\left(\mv I + \left(\mv F\mv S_E\mv F^H + \sigma^2 \mv I\right)^{-1}\mv H \mv S_I \mv H^H\right).
\end{align}

\section{Optimal Design with Multi-Beam WET}\label{sec:multi_beam}

In this section, we study the optimal transmit signals design in the coexisting (primary) WET and (secondary) WIT systems. Here, we consider that the WET system is oblivious to the WIT system and designs the transmit energy signal at the ET independently (i.e. without the need to minimize the interference to the IR); while the WIT system adjusts the transmit information signal at the IT subject to the interference from the ET.

In the WET system, to balance between the efficiency and user fairness of energy transfer, we maximize the total energy transferred to all the $K$ ERs over the whole block, subject to the energy fairness constraint that is specified based on the concept of {\it energy-profile}, similar to the rate-profile concept proposed in \cite{Mohseni:RateProfile}. Mathematically, we formulate the following optimization problem with a particular energy-profile vector $\mv{\alpha} \triangleq (\alpha_1, \ldots, \alpha_K)^T$:
\begin{align*}
{\rm{(P1)}}:\max\limits_{\mv S_E, \Theta}~&\Theta\\
{\rm{s.t.}}~& T\mathrm{tr}\left(\mv G_k^H\mv G_k\mv S_E\right) \ge \alpha_k \Theta,~ \forall k\in\{1,\ldots,K\}\\
&\mv S_E \succeq \mv 0, \mathrm{tr}\left(\mv S_E\right) \le P_E,
\end{align*}
where $\alpha_k \ge 0, k\in\{1,\ldots,K\}$, denotes the target ratio of the $k$th ER's harvested energy to the total harvested energy by all ERs, given by $\Theta$, with $\sum_{k=1}^K \alpha_k = 1$. Note that $\mv{\alpha}$ is a parameter designed based on the energy requirements among different ERs. It can be shown that (P1) is a convex semi-definite program (SDP) \cite{ConvexOptimization}, and thus can be solved by standard convex optimization techniques such as CVX \cite{cvx}. Let the optimal solution to (P1) be denoted as $\mv S^*_E$ and $\Theta^*$. Then, we have the following two important properties for $\mv S^*_E$.
\begin{itemize}
  \item It holds that $\mathrm{tr}(\mv S^*_E) = P_E$, i.e., all the available transmit power should be used up by the ET to maximize the energy transferred to all the ERs.
  \item When the number of ERs $K$ becomes large, it follows that $d^*_E = \mathrm{rank}\left(\mv S^*_E\right) > 1$ in general \cite{LiuZhangChia}, i.e., more than one energy beams are required at the optimal solution for balancing the energy fairness among different ERs.
\end{itemize}

Next, given the transmit energy covariance matrix $\mv S_E^*$, we design the transmit information covariance matrix $\mv S_I$ at the IT to maximize the achievable rate at the IR, given by $R(\mv S_E^*,\mv S_I)$ in (\ref{eqn:info:rate}). Accordingly, this problem is formulated as
\begin{align*}
(\mathrm{P2}):~\max_{\mv S_I}~&R(\mv S_E^*,\mv S_I)\\
{\rm s.t.}~& \mv S_I \succeq \mv 0,~ \mathrm{tr}(\mv S_I) \le P_I.
\end{align*}
It is evident that problem (P2) is equivalent to the conventional rate maximization problem for a point-to-point MIMO channel in \cite{Telatar1999}, by considering $\left(\mv F\mv S_E^*\mv F^H + \sigma^2 \mv I\right)^{-1/2}\mv H$ as the equivalent MIMO channel matrix in (\ref{eqn:info:rate}). Then, the optimal solution to (P2), denoted by $\mv S_I^*$, can be obtained by performing singular value decomposition (SVD) on $\left(\mv F\mv S_E^*\mv F^H + \sigma^2 \mv I\right)^{-1/2}\mv H$ together with a water-filling power allocation \cite{Telatar1999}. Note that to practically obtain such an optimal solution, the IT requires to know the interference-plus-noise covariance matrix $\mv F\mv S_E^*\mv F^H + \sigma^2 \mv I$, which can be practically estimated by the IR and sent back to the ET. Then, the maximum achievable rate of the WIT system is given by $R(\mv S_E^{*},\mv S^{*}_{I})$.

Now, it is interesting to analyze $R(\mv S_E^{*},\mv S^{*}_{I})$ to show the impact from the WET system (accordingly, the resulted one-way interference) to the throughput performance of the coexisting WIT system. We are particularly interested in the pre-log factor of the achievable rate (also known as the DoF or the multiplexing gain) of the WIT system in the high signal-to-noise ratio (SNR) regime with $P_I \to\infty$ and $P_E = \alpha P_I \to\infty$.
\begin{proposition}\label{proposition:DoFl}
As $P_I \to\infty$ and $P_E = \alpha P_I\to\infty$, it follows that ${R({\mv S}_E^*,{\mv S}_I^*)}/{\log_2(P_I)} = \min(M_I,\max(N_I - d_E^*,0))$ with $d_E^* = \mathrm{rank}(\mv S_E^*)$ being the number of energy beams sent by the ET.
\end{proposition}
\begin{proof}
With $d_E^*$ energy beams sent by the ET, its resulting interference signal to the IR satisfies $\mathrm{rank}\left(\mv F\mv S_E^*\mv F^H\right) = \min(N_I,d_E^*)$, since $\mv F$ is of full-rank and $\mv S_E^*$ (or $\mv G_k$'s) and $\mv F$ are independent. As a result, there are in total $N_I - d_E^*$ linearly independent basis vectors in the interference-free signal space, provided that the IR has $N_I\geq d_E^*$ receive antennas; therefore, the IR can support a total of $\max(N_I - d_E^*,0)$ DoF \cite{CadambeJafar2008}. By using this together with the fact that the IT has $M_I$ transmit antennas, we have that the DoF of the WIT system is indeed $\min(M_I,\max(N_I - d_E^*,0))$, provided that the corresponding channel matrix $\mv H$ is also of full-rank and is independent of $\mv S_E^*$ and $\mv F$. This completes the proof of this proposition.
\end{proof}

It is observed from Proposition \ref{proposition:DoFl} that the DoF of the WIT system, i.e., $\min(M_I,\max(N_I - d_E^*,0))$, critically depends on the number of energy beams $d_E^*$ in the WET system. To achieve the maximum DoF for WIT, the ET should minimize the number of energy beams transmitted in space.

\section{Alternative Design with Single-Beam WET}\label{sec:single_beam}

In this section, we propose an alternative single-beam WET scheme by the ET sending only one energy beam with adjustable weights over time, so as to achieve the same optimal performance in the WET system and the maximum DoF in the WIT system at the same time.

%and sends one of the $d_E^*$ spatially multiplexed energy beams above at each sub-block. Given , we denote $\gamma_1^*, \ldots, \gamma_{d_E^*}^*$ as its $d^*_E$ strictly positive eigenvalues, and ${\mv w}_1^*, \ldots, {\mv w}^*_{d_E^*}$ as the corresponding $d^*_E$ eigenvectors (each in fact corresponding to one (normalized) energy beam), where $\sum_{i=1}^{d_E^*}\gamma_i^* = P_E$ due to $\mathrm{tr}(\mv S_E^*) = P_E$. Then, the transmission block of interest is divided into $d^*_E$ sub-blocks

Specifically, this new single-beam WET scheme is designed based on the time-sharing among the $d_E^*$ optimal energy beams obtained by solving (P1), which are specified by the optimal transmit energy covariance matrix $\mv S^*_E$. Let $\gamma_1^*, \ldots, \gamma_{d_E^*}^*$ denote the $d^*_E$ strictly positive eigenvalues of $\mv S^*_E$, and ${\mv w}_1^*, \ldots, {\mv w}^*_{d_E^*}$ denote their corresponding eigenvectors, where $\mv S_E^* = \sum_{i=1}^{d_E^*}\gamma_i^*{\mv w}^{*}_{i}{\mv w}^{*H}_{i}$ and $\sum_{i=1}^{d_E^*}\gamma_i^* = P_E$ (due to  $\mathrm{tr}(\mv S^*_E) = P_E$). Note that $\{{\mv w}^*_i\}_{i=1}^{d_E^*}$ are the $d_E^*$ optimal energy beams corresponding to $\mv S^*_E$. Then, the ET divides the whole transmission block into $d_E^*$ sub-blocks each having a length of $t_i = {\gamma_i^*T}/P_E, i\in\{1,\ldots,d^*_E\}$, where $\sum_{i=1}^{d^*_E} t_i = T$. Over each sub-block $i$, the ET uses the full transmit power $P_E$ to send the $i$th optimal energy beam (i.e., ${\mv w}_i^*$), with the transmit energy covariance matrix given by $\mv S_{E,i}^\star = P_E {\mv w}_i^* {\mv w}_i^{*H}$, where $\mathrm{rank}(\mv S_{E,i}^\star) = 1, i\in\{1,\ldots,d_E^*\}$. Therefore, the harvested energy by ER $k$ at the $i$th sub-block is expressed as
\begin{align}\label{eqn:energy:A}
Q_{k,i}&(\{\mv S_{E,i}^\star\}) = t_i\mathrm{tr}\left(\mv G_k^H\mv G_k\mv S_{E,i}^\star\right) \nonumber\\
&= {\gamma_i^*T} \mathrm{tr}\left(\mv G_k^H\mv G_k{\mv w}_i^* {\mv w}_i^{*H}\right), k\in\{1,\ldots,K\}.
\end{align}
By combining the $d_E^*$ sub-blocks, the total harvested energy by ER $k$ over the whole block is $\sum_{i=1}^{d_E^*} Q_{k,i}(\{\mv S_{E,i}^\star\})$. Then, we have the following proposition.
\begin{proposition}\label{proposition:4.1}
The alternative single-beam WET scheme achieves the same harvested energy at each ER $k$ over the whole block as the optimal multi-beam WET scheme with $\mv S_E^*$, i.e., $\sum_{i=1}^{d_E^*} Q_{k,i}(\{\mv S_{E,i}^\star\})=Q_k(\mv S_{E}^*), \forall k\in\{1,\ldots,K\}$.
\end{proposition}
\begin{proof}
This proposition  can be proved via simple manipulations by using $\mv S_E^* = \sum_{i=1}^{d_E^*} \gamma_i^*{\mv w}_i^{*} {\mv w}_i^{*H}$ together with (\ref{eqn:energy}) and (\ref{eqn:energy:A}). Thus, the details are omitted.
\end{proof}

Proposition \ref{proposition:4.1} is somewhat surprising, but can be intuitively explained as follows. Note that the new single-beam WET scheme indeed employs the same $d_E^*$ energy beams (via time sharing) as those in the optimal multi-beam WET scheme (via spatial multiplexing). Since the harvested power at each ER (see (\ref{eqn:energy}) and (\ref{eqn:energy:A})) is a linear function with respect to the transmit energy covariance matrix at the ET, the two schemes achieve the same optimal WET performance for all the ERs.

Next, we consider the WIT system. Since the one-way interference from the WET system varies over sub-blocks (due to the different transmit energy covariance matrix employed at each sub-block), the WIT system should correspondingly adjust the information signals at the IT for each of the $d_E^*$ sub-blocks. For convenience, we assume that the information signal adjustment at the IT is perfectly synchronized with the energy signal adaptation at the ET. Let $\mv{S}_{I, i}$ denote the transmit information covariance matrix at the IT in the $i$th sub-block, $i\in\{1,\ldots,d_E^*\}$. Then the achievable rate of the WIT system (in bps/Hz) over the $i$th sub-block is expressed as $R(\mv S^{\star}_{E,i},\mv S_{I,i})$ in (\ref{eqn:info:rate}), and the average rate over the whole block is given by $\frac{1}{T} \sum_{i=1}^{d_E^*} t_i R(\mv S^{\star}_{E,i},\mv S_{I,i})$. As a result, the average rate optimization problem for the WIT system is formulated as
\begin{align}
{\mathrm{(P3)}}: \max_{\{ \mv S_{I,i}\}}~&\frac{1}{T} \sum_{i=1}^{d_E^*} t_i R(\mv S^{\star}_{E,i},\mv S_{I,i}) \nonumber\\
\mathrm{s.t.}~& \mv S_{I,i} \succeq \mv 0,~ \mathrm{tr}(\mv S_{I,i}) \le P_I, \forall i \in \{1,\ldots,d_E^{*}\}. \nonumber
\end{align}
Problem (P3) can be decomposed into $d_E^*$ sub-problems each for one sub-block, which can then be solved similarly as (P2). Let the optimal solution to (P3) be denoted by $\{\mv S_{I,i}^\star\}$. Accordingly, we denote the maximum average rate of the WIT system over the whole block as $\frac{1}{T} \sum_{i=1}^{d_E^*} t_i R(\mv S^{\star}_{E,i},\mv S_{I,i}^\star)$. We have the following proposition.
\begin{proposition}\label{proposition:DoF:SB}
As $P_I \to\infty$ and $P_E = \alpha P_I \to\infty$,  it follows that  $\frac{1}{T} \sum_{i=1}^{d_E^*} t_i R(\mv S^{\star}_{E,i},\mv S_{I,i}^\star)/\log_2(P_I) = \min(M_I,\max(N_I - 1,0))$.
\end{proposition}
\begin{proof}
This proposition follows from Proposition \ref{proposition:DoFl} together with $\mathrm{rank}(\mv S_{I,i}^\star) = 1, \forall i \in\{1,\ldots, {d_E^*}\}$.
\end{proof}

By comparing Propositions \ref{proposition:DoF:SB} and \ref{proposition:DoFl}, it is evident that as long as $d_E^* > 1$, $M_I > 1$ and $N_I > 1$, the WIT system under the single-beam WET scheme here can achieve higher DoF than that under the multi-beam WET scheme in the previous section. Nevertheless, when $M_I  = 1$ or $N_I = 1$, there is no DoF gain for the alternative design with single-beam WET. Despite this, we show in the following proposition that under the general case with any arbitrary number of $M_I$ and $N_I$ (including $M_I  = 1$ or $N_I = 1$) and any transmit power values of $P_I$ and $P_E$, the design with single-beam WET here is still beneficial over the previous design with multi-beam WET, in terms of the achievable rate of the coexisting WIT system.

\begin{proposition}\label{proposition:comparison}
It follows that $\frac{1}{T} \sum_{i=1}^{d_E^*} t_i R(\mv S^{\star}_{E,i},\mv S_{I,i}^\star) \ge R(\mv S^{*}_{E},\mv S_{I}^{*})$.
\end{proposition}
\begin{proof}
Note that it can be verified that $\frac{1}{T} \sum_{i=1}^{d_E^*} t_i \mv S^{\star}_{E,i} = \mv S^{*}_{E}$ via some simple manipulations. Then, it follows that $\frac{1}{T} \sum_{i=1}^{d_E^*} t_i R(\mv S^{\star}_{E,i},\mv S_{I}^*) \ge R(\mv S^{*}_{E},\mv S_{I}^{*})$, since it can be shown that $R(\mv S_{E},\mv S_{I})$ is a convex function with respect to $\mv S_{E}$ under given $\mv S_I$. In addition, we have $\frac{1}{T} \sum_{i=1}^{d_E^*} t_i R(\mv S^{\star}_{E,i},\mv S_{I,i}^\star) \ge \frac{1}{T} \sum_{i=1}^{d_E^*} t_i R(\mv S^{\star}_{E,i},\mv S_{I}^*)$, since $\{\mv S_{I,i}^\star\}$ is optimal for problem (P3). By combining the above two facts, this proposition is verified.
\end{proof}

\section{Numerical Results}\label{sec:Numerical}

\begin{table}[!t]\scriptsize
\caption{Results on The Number of Energy Beams $d_E^*$}
\label{table1} \centering
\begin{tabular}{|p{0.5in}|p{0.5in}|p{0.4in}|p{0.4in}|p{0.4in}|}
\hline
&{$d_E^*=1$}& {$d_E^*=2$}&{$d_E^*=3$}&{{ $d_E^*=4$}} \\ \hline
{$K=10$}&{251}& {747}  & {2}   & {{0} }\\ \hline
{$K=20$}&{2}& {679}  & {319}   & {{0} }\\ \hline
{$K=40$}&{0}& {   141}  & {   823}   & {36} \\ \hline
\end{tabular}
\end{table}

In this section, we provide simulation results to validate our studies above. We assume that the ERs are located at an identical distance of 5 meters from the ET, for which the average path loss from the ET to each ER is 40 dB; while the distances from the IT and the ET to the IR are the same of 30 meters, for which the average path loss are both 80 dB. Rican fading channel models are considered for the MIMO links from the ET to each ER \cite{LiuZhangChia}, while Rayleigh fading channel models are used for the other links. We set the number of transmit antennas at the ET as $M_E=4$, the number of receive antennas at each ER as $N_E=1$, the number of transmit antennas at the IT as $M_I = 4$, the energy harvesting efficiency at each ER as $\eta = 50\%$, and the noise power at the IR as $\sigma^2 = -70$ dBm. We also consider that $\alpha_1 = \cdots=\alpha_K = 1/K$, such that each ER can harvest the same amount of energy.

First, consider the WET system. Table \ref{table1} shows the number of the energy beams $d_E^*$ obtained from the optimal solution to (P1), where 1000 random channel realizations are considered and the transmit power at the ET is set as $P_E=30$ dBm ($1$ W). It is observed that as the number of ERs $K$ increases, generally more energy beams are required to balance the energy fairness among ERs. For example, in 823 among the 1000 realizations, we have $d_E^*=3$ when $K=40$. In addition, the average harvested energy at each ER is computed to be $\Theta^* = 0.0571$ mW, $0.0429$ mW, and $0.0349$ mW in the cases with $K=10, 20$, and $40$, respectively. This shows that a larger $K$ value results in less harvested energy at each individual ER, since in this case the ET needs to more uniformly distribute its transmit power to these ERs (i.e., less energy beamforming gain is achievable).

\begin{figure}
\centering
 \epsfxsize=1\linewidth
    \includegraphics[width=10cm]{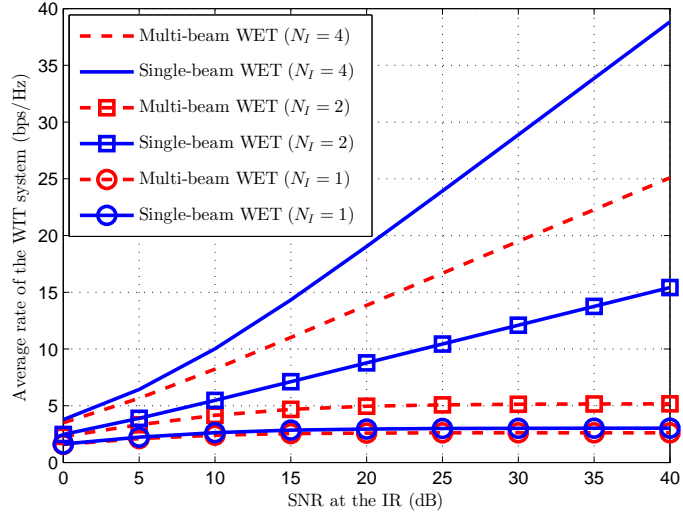}
\caption{The average rate of the WIT system versus the SNR at the IR.} \label{fig:3}
\end{figure}

Next, consider the coexisting WIT system. Fig. \ref{fig:3} shows its average rate versus the SNR at the IR with $P_E = P_I$ and $K=20$, where the SNR (in dB) is defined as $\mathrm{SNR} = P_I-80$dB$- {\sigma^2} = P_I-10$dBm. It is observed that in the cases with $N_I=2$ and $N_I=4$, the WIT system under the single-beam WET scheme achieves higher DoF than that under the multi-beam WET scheme. This is expected and can be explained based on Propositions \ref{proposition:DoFl} and \ref{proposition:DoF:SB}, provided that the number of energy beams employed at the ET $d_E^*$ is normally larger than one in the case of $K=20$ (cf. Table \ref{table1}). When $N_I=1$, although the WIT system becomes interference limited (with the DoF being zero) under both WET schemes, the one with the single-beam WET scheme is still observed to have a higher average rate than that with the multi-beam WET scheme. This is consistent with our analysis in Proposition \ref{proposition:comparison}.

\section{Conclusion}

This letter investigates spectrum sharing between a multiuser MIMO WET system and a point-to-point MIMO WIT system. In such a scenario, the conventional multi-beam design in the WET system causes high-rank interference and leads to severe performance degradation at the coexisting WIT system. To address this issue, we propose a new single-beam WET scheme, where the ET sends only one energy beam with adjustable weights over time. This new design achieves the same optimal performance for the WET system and significantly reduces the harmful interference to the WIT system as compared to the multi-beam WET scheme. Our results provide new insights on the energy beamforming design in MIMO WET systems and minimizing their adverse impact to coexisting opportunistic communications.

\end{document}